\documentclass[oneside,american]{amsart}
\usepackage{mathptmx}

\usepackage[T1]{fontenc}
\usepackage[latin9]{inputenc}
\usepackage{verbatim}
\usepackage{mathrsfs}
\usepackage{amsbsy}
\usepackage{amsthm}
\usepackage{amssymb}
\usepackage{graphicx}

\makeatletter
\numberwithin{equation}{section}
\numberwithin{figure}{section}
\theoremstyle{plain}
\newtheorem{thm}{\protect\theoremname}
\theoremstyle{remark}
\newtheorem{rem}[thm]{\protect\remarkname}
\theoremstyle{plain}
\newtheorem{lem}[thm]{\protect\lemmaname}
\theoremstyle{plain}
\newtheorem{cor}[thm]{\protect\corollaryname}

\makeatother

\usepackage{babel}
\providecommand{\corollaryname}{Corollary}
\providecommand{\lemmaname}{Lemma}
\providecommand{\remarkname}{Remark}
\providecommand{\theoremname}{Theorem}

\begin{document}
\title[On computing bound states of the Dirac and Schr\"odinger Equations]{On computing bound states of the Dirac and Schr\"odinger Equations
}
\author{Gregory Beylkin$\,^{*}$, Joel Anderson$\,^{**}$ and Robert J. Harrison$\,^{**}$}
\address{$\,^{*}$Department of Applied Mathematics\\
University of Colorado at Boulder \\
UCB 526, Boulder, CO 80309-0526\\
~\\
$\,^{**}$Institute for Advanced Computational Science\\
Stony Brook University\\
Stony Brook, NY 11794-5250}
\begin{abstract}
We cast the quantum chemistry problem of computing bound states as
that of solving a set of auxiliary eigenvalue problems for a family
of parameterized compact integral operators. The compactness of operators
assures that their spectrum is discrete and bounded with the only
possible accumulation point at zero. We show that, by changing the
parameter, we can always find the bound states, i.e., the eigenfunctions
that satisfy the original equations and are normalizable. While for
the non-relativistic equations these properties may not be surprising,
it is remarkable that the same holds for the relativistic equations
where the spectrum of the original relativistic operators does not
have a lower bound. We demonstrate that starting from an arbitrary
initialization of the iteration leads to the solution, as dictated
by the properties of compact operators.
\end{abstract}

\thanks{The authors would like to thank Stony Brook Research Computing and
Cyberinfrastructure, and the Institute for Advanced Computational
Science (IACS) at Stony Brook University (SBU) for access to the high-performance
SeaWulf computing system, which was made possible by a \$1.4M National
Science Foundation grant (\#1531492). Anderson acknowledges financial
support from IACS at SBU, and Harrison acknowledges support from the
National Science Foundation grant OAC-1931387.}

\maketitle
\section{Introduction}

We develop, analyze, and demonstrate a method for computing bound
states for both the non-relativistic Schr\"odinger and relativistic
Dirac equations of quantum chemistry (QC). As is well-known, the equations
of quantum mechanics allow both bound and scattering states; \textit{inter
alia} this implies that, if the scattering states are present, then
it is impossible to construct a single self-adjoint operator that
would have only the bound states since the resulting basis of eigenvectors
would be incomplete. We show that, instead, there is an auxiliary
one parameter family of compact operators such that the sought bound
states are found by computing eigenfunctions of particular members
of this family. In fact, this family of operators constructed using
a parameterized Green's function and yielding (what we call) the integral
form of the Schr\"odinger equation is well-known as it has been used
in mathematical analysis since 1950s. It has long been established
that, for certain classes of potentials and for the parameter corresponding
to the bound states, the relevant operators are compact. A particularly
important for QC problems is the Rollnik class of potentials arrived
at and analyzed by several authors, see \cite{ROLLNI:1956,SCHWIN:1961,GRO-WU:1961,GRO-WU:1962,SC-WE-WR:1964,SIMON:1971,SIMON:2005}
and references therein. We elaborate on this later in the paper.

The early use of the integral form of the Schr\"odinger equation
for computational purposes was pioneered by Kalos using Monte-Carlo
approach \cite{KALOS:1963} and, more recently, employed in adaptive
multiresolution QC algorithms in \cite{H-F-Y-G-B:2004,A-S-H-B:2019}
(see also \cite{H-B-B-C-F-F-G-etc:2016} and references therein).
In this paper we revisit our approach for computing bound states by
solving the integral form of the Dirac equations of QC in \cite{A-S-H-B:2019}
(and that of the non-relativistic equations in \cite{H-F-Y-G-B:2004})
to demonstrate that convergence interpreted in these papers as being
local (thus requiring an initial guess sufficiently close to the desired
bound state eigenfunction) is, in fact, a robust global convergence.
We demonstrate this by showing that relevant operators are compact.

The integral form of the Schr\"odinger equation is obtained using
a parameterized Green's function yielding an auxiliary family of integral
operators. The parameter can be tuned and, at certain discrete values,
the corresponding operator yields the desired bound states. This approach
found its application in an adaptive multiresolution method for solving
non-relativistic equations of QC for computing bound states introduced
in \cite{H-F-Y-G-B:2004} and implemented in MADNESS (Multiresolution
ADaptive Numerical Environment for Scientific Simulation, see \cite{H-B-B-C-F-F-G-etc:2016}).
We note that since the spectrum of non-relativistic Hamiltonians of
QC is bounded from below, it is not surprising that the spectra of
all operators of the auxiliary family of integral operators are also
bounded from below as they turn out to be compact operators for all
values of the parameter. While it may appear that using a family of
integral operators instead of the single Hamiltonian is computationally
more expensive, the separated multiresolution representation of operators
in \cite{H-F-Y-G-B:2004} actually makes MADNESS a fast method for
accurate computations in QC and nuclear sciences \cite{H-B-B-C-F-F-G-etc:2016}.

Turning to the relativistic equations, it is well-known that the spectrum
of the Dirac Hamiltonian is not bounded from below so that the problem
of finding bound states can not be cast as that of minimization of
a quadratic form, the so-called Rayleigh quotient. As a consequence,
computing bound states by directly discretizing the relativistic Hamiltonian
requires additional numerical devices to guide computation to the
desired bound states. Thus, it is remarkable that for relativistic
equations of QC the associated auxiliary family of integral operators
consists of compact operators as long as the parameter --- the relativistic
energy --- is selected to be positive. We show that, by selecting
the positive relativistic energy, the equations yield only the desired
bound states so that there is no need to numerically exclude either
the positive-energy scattering states or all negative-energy states.

The integral form of the Dirac Hamiltonian was proposed in \cite{BLA-BAB:2013}
and used in \cite{A-S-H-B:2019} for computing relativistic energies
of bound states. In that context, we had originally interpreted convergence
of the integral iteration as being local --- i.e., requiring an initial
guess sufficiently close to the desired bound state eigenfunction.
Instead, in this paper we demonstrate, both theoretically and numerically,
robust global convergence to the sought bound state. We cast the iteration
that yields the desired bound states as a combination of the power
method and the Newton's method, which allows us to demonstrate that
we can always find an appropriate member of the family of integral
operators in order to obtain the desired bound state (in practical
computations these two iterations are usually to be combined to gain
speed).

We start by briefly describing the relevant one-electron non-relativistic
Schr\"odinger and relativistic Dirac equations of QC, and the corresponding
families of auxiliary parameterized integral operators. Next, we consider
the Hilbert spaces in which we seek solutions and discuss under which
conditions on the potentials the operators of interest are Hilbert-Schmidt
operators. We then show that the integral operators of the auxiliary
family of operators for both, non-relativistic and relativistic problems,
are Hilbert-Schmidt operators and, therefore, compact. Finally we
discuss an iteration to select a particular member of the family of
auxiliary operators in order to solve the original problem of finding
a bound state. Some of the proofs needed in the main text can be found
in Appendices~\ref{subsec:Appendix-A:-Compact A^*A}-\ref{subsec:Appendix-D:-Dependence of eig relativistic}.
\section{Integral form of equations of quantum chemistry}

We use examples of the one-electron Kohn-Sham \cite{PAR-YAN:1989}
equations and their relativistic counterparts in Dirac's formulation
with the point-charge Coulomb potential \cite{LIU:2010} to motivate
and demonstrate our approach. Methods of quantum chemistry differ
in how they replace the electron-electron interaction by the interaction
of an electron with an averaged field generated by all electrons;
however the singular part of the total potential, the Coulomb electron-nuclear
interaction, is the same in most methods \cite{LIU:2010,LIU:2020}.
The singularity does not arise if a finite-size model is adopted for
the nuclear charge distribution \cite{VIS-DYA:1997} instead of the
idealized point-charge model, although any smoothing only occurs within
the radius of the nuclear charge distribution that is $O\left(10^{-5}\right)$
smaller than the radius of the atomic charge distribution. Relativistic
pseudo-potentials \cite{DOL-CAO:2012} further smooth the potential
and also eliminate the most tightly-bound (i.e., core) electrons that
experience the strongest relativistic effects. However, all-electron
calculations with both point and finite nuclear models are still essential
in order to directly access the properties of core electrons in molecular
environments, to employ general relativistic Hamiltonians, and to
eliminate/assess the approximations inherent to pseudo-potentials.
In demonstrating compactness of a family of integral operators, the
Coulomb electron-nuclear interaction presents the main obstacle and,
for this reason, we can limit our discussion to these two examples
of QC equations.

\subsection{The Kohn-Sham equations}

Consider occupied orbitals $\psi_{i}\left(\mathbf{r}\right)$, $i=1,\dots,N$
defining the electron density 
\[
\rho\left(\mathbf{r}\right)=2\sum_{i=1}^{N}\left|\psi_{i}\left(\mathbf{r}\right)\right|^{2},
\]
which are the lowest $N$ eigenfunctions of the Kohn-Sham operator
\begin{equation}
\left(-\frac{1}{2}\Delta+V\left(\mathbf{r}\right)\right)\psi_{i}\left(\mathbf{r}\right)=E_{i}\psi_{i}\left(\mathbf{r}\right),\,\,\,i=1,\dots.N,\label{eq:Kohn-Sham}
\end{equation}
where 
\begin{equation}
V\left(\mathbf{r}\right)=V_{ext}\left(\mathbf{r}\right)+V_{c}\left(\mathbf{r}\right)+V_{xc}\left(\mathbf{r}\right).\label{eq: components of the potential}
\end{equation}
For molecules, the external potential includes the attraction of the
electrons to the nuclei,
\begin{equation}
V_{ext}\left(\mathbf{r}\right)=-\sum_{\alpha}\frac{Z_{\alpha}}{\left\Vert \mathbf{r}-\mathbf{R}_{\alpha}\right\Vert },\label{eq:nuclearpotential}
\end{equation}
($Z_{\alpha}$ and $\boldsymbol{R}_{\alpha}$being the nuclear charge
and position, respectively). The Coulomb potential describes the repulsion
between electrons, 
\[
V_{c}\left(\mathbf{r}\right)=\int_{\mathbb{R}^{3}}\frac{\rho\left(\mathbf{r}'\right)}{\left\Vert \mathbf{r}-\mathbf{r}'\right\Vert }d\mathbf{r}'
\]
and the exchange-correlation potential $V_{xc}$ that in this work
is taken to be a scalar that within the generalized gradient approximation
(GGA) depends on $\rho$ and its derivatives at that point. Non-local
potentials, such as the Hartree-Fock exchange potential \cite{PAR-YAN:1989},
can be included as well. The Hartree-Fock exchange potential has singularities
at locations of nuclei (e.g. cusps in the non-relativistic case),
but these singularities are necessarily weaker than those of the external
potential $V_{ext}$ at the same locations. Our estimates depend only
on the slowest decay exhibited by the total potential in momentum
space, which are due to the stronger singularity of $V_{ext}$.

Introducing the Green's function
\begin{equation}
\left(-\Delta+\mu^{2}\right)G_{\mu}\left(\mathbf{r},\mathbf{r}'\right)=\delta\left(\mathbf{r}-\mathbf{r}'\right),\,\,\,\,\,\,G_{\mu}\left(\mathbf{r},\mathbf{r}'\right)=\frac{e^{-\mu\left\Vert \mathbf{r}-\mathbf{r}'\right\Vert }}{4\pi\left\Vert \mathbf{r}-\mathbf{r}'\right\Vert },\label{eq:Green's function, non-rel case}
\end{equation}
we consider the auxiliary coupled eigenvalue problems
\begin{equation}
\lambda_{\mu_{i}}\psi_{\mu_{i}}=-2G_{\mu_{i}}V\psi_{\mu_{i}},\,\,\,\,i=1,\dots,N.\label{eq:auxiliary eigenvalue problem Kohn-Sham}
\end{equation}
If $\mu_{i}^{2}=-2E_{i}$, then $\lambda_{\mu_{i}}=1$ and functions
$\psi_{\mu_{i}}$ also solve (\ref{eq:Kohn-Sham}). Introducing functions
\[
\phi_{\mu_{i}}=G_{\mu_{i}}^{-1/2}\psi_{\mu_{i}},
\]
the auxiliary eigenvalue problem (\ref{eq:auxiliary eigenvalue problem Kohn-Sham})
becomes 
\begin{equation}
\lambda_{\mu_{i}}\phi_{\mu_{i}}=-2G_{\mu_{i}}^{1/2}VG_{\mu_{i}}^{1/2}\phi_{\mu_{i}},\,\,\,\,i=1,\dots,N.\label{eq:auxiliary eigenvalue problem Kohn-Sham-1}
\end{equation}
In what follows, we consider operators $VG_{\mu_{i}}^{1/2}$ and show
that, by choosing an appropriate Hilbert space for the solutions (\ref{eq:auxiliary eigenvalue problem Kohn-Sham}),
these operators are compact.
\begin{rem}
\label{rem: neg.definite, non-relativistic}In the analysis that follows,
we want the potential $V$ in (\ref{eq:auxiliary eigenvalue problem Kohn-Sham})
(and therefore in (\ref{eq:auxiliary eigenvalue problem Kohn-Sham-1}))
to be a negative definite multiplication operator. While $V_{ext}$
is negative, the additional components of the total potential, $V_{c}\left(\mathbf{r}\right)$
and $V_{xc}\left(\mathbf{r}\right)$, may possibly violate this property
of $V$ in some subdomains. In such case, as long as $V_{c}\left(\mathbf{r}\right)$
and $V_{xc}\left(\mathbf{r}\right)$ are bounded, we can modify the
derivation of (\ref{eq:auxiliary eigenvalue problem Kohn-Sham}) and
(\ref{eq:auxiliary eigenvalue problem Kohn-Sham-1}) by first shifting
the spectrum 
\begin{equation}
\left(-\frac{1}{2}\Delta+V\left(\mathbf{r}\right)-\tau\right)\psi_{i}\left(\mathbf{r}\right)=\left(E_{i}-\tau\right)\psi_{i}\left(\mathbf{r}\right),\,\,\,i=1,\dots.N,\label{eq:Kohn-Sham-shifted}
\end{equation}
where $\tau>0$ is a sufficiently large shift so that
\[
V\left(\mathbf{r}\right)-\tau<0,
\]
and considering 
\begin{equation}
\lambda_{\mu_{i}}\psi_{\mu_{i}}=-2G_{\mu_{i}}\left(V-\tau\right)\psi_{\mu_{i}}.\label{eq:shifted integral eq}
\end{equation}
While the shifted potential is not zero at infinity, for as long as
the eigenfunctions $\psi_{\mu_{i}}$ of (\ref{eq:Kohn-Sham}) decay
exponentially (as it is the case for the Schr\"odinger equation,
see \cite{AGMON:1985}), the integrals in (\ref{eq:shifted integral eq})
are well defined. In fact, in solving (\ref{eq:Kohn-Sham}) numerically,
the functions $\psi_{\mu_{i}}$ are considered to be non-zero only
in a bounded domain. In (\ref{eq:shifted integral eq}), if $\mu_{i}^{2}=-2\left(E_{i}-\tau\right)$,
then $\lambda_{\mu_{i}}=1$ and functions $\psi_{\mu_{i}}$ also solve
(\ref{eq:Kohn-Sham}). So far in our computations we did not encounter
a need to shift the spectrum. However, without loss of generality,
it is important to consider $V$ to be a negative definite multiplication
operator.
\end{rem}

\subsection{Integral form of Dirac's equations}

An orbital (one-particle eigenfunction) of the relativistic Dirac
equations is a four-component vector-function (i.e., a spinor) $\psi$
which satisfies
\begin{equation}
\mathcal{H}\psi=E\psi\label{eq:eigenvalue problem Dirac's eq}
\end{equation}
where
\begin{equation}
\mathcal{H}=\mathcal{H}_{0}+\mathcal{V}\left(\mathbf{r}\right),\label{eq:Dirac hamiltonian}
\end{equation}
\[
\mathcal{H}_{0}=\frac{\hbar c}{i}\left(\alpha_{1}\frac{\partial}{\partial x_{1}}+\alpha_{2}\frac{\partial}{\partial x_{2}}+\alpha_{3}\frac{\partial}{\partial x_{3}}\right)+\beta mc^{2},
\]
\[
\alpha_{1}=\left(\begin{array}{cc}
\mathbf{0} & \sigma_{1}\\
\sigma_{1} & \mathbf{0}
\end{array}\right),\,\,\,\alpha_{2}=\left(\begin{array}{cc}
\mathbf{0} & \sigma_{2}\\
\sigma_{2} & \mathbf{0}
\end{array}\right),\,\,\,\alpha_{3}=\left(\begin{array}{cc}
\mathbf{0} & \sigma_{3}\\
\sigma_{3} & \mathbf{0}
\end{array}\right)\,\,\,\beta=\left(\begin{array}{cc}
\sigma_{0} & \mathbf{0}\\
\mathbf{0} & -\sigma_{0}
\end{array}\right),
\]
and

\[
\sigma_{0}=\left(\begin{array}{cc}
1 & 0\\
0 & 1
\end{array}\right),\,\,\,\sigma_{1}=\left(\begin{array}{cc}
0 & 1\\
1 & 0
\end{array}\right),\,\,\,\sigma_{2}=\left(\begin{array}{cc}
0 & -i\\
i & 0
\end{array}\right),\,\,\,\sigma_{3}=\left(\begin{array}{cc}
1 & 0\\
0 & -1
\end{array}\right),
\]
are the Pauli matrices such that 
\[
\sigma_{1}^{2}=\sigma_{2}^{2}=\sigma_{3}^{2}=\sigma_{0},
\]
and 
\[
\sigma_{1}\sigma_{2}=i\sigma_{3},\,\,\,\sigma_{2}\sigma_{3}=i\sigma_{1},\,\,\,\sigma_{3}\sigma_{1}=i\sigma_{2}.
\]
The matrix-potential operator is 
\begin{equation}
\mathcal{V}\left(\mathbf{r}\right)=\left(\begin{array}{cccc}
V & 0 & 0 & 0\\
0 & V & 0 & 0\\
0 & 0 & V & 0\\
0 & 0 & 0 & V
\end{array}\right),\label{eq:relativistic potential}
\end{equation}
where the potential $V$ has several components as in \ref{eq: components of the potential}.
As in the non-relativistic case, we can consider a system of equations
involving $N$ orbitals where the electron-electron interactions are
captured as an interaction of an electron with an averaged potential
generated by all electrons. As in the non-relativistic Kohn-Sham equations,
we can assign the four-component orbitals to each electron and consider
a system of coupled eigenvalue problems. Avoiding complicating notations,
without loss of generality, we only consider operators that involve
a single orbital.

Setting 
\[
\mathcal{I}=\left(\begin{array}{cc}
\sigma_{0} & \mathbf{0}\\
\mathbf{0} & \sigma_{0}
\end{array}\right),
\]
we rewrite (\ref{eq:eigenvalue problem Dirac's eq}) as
\begin{equation}
\left(\mathcal{H}_{0}-E\mathcal{I}\right)\psi=-\mathcal{V}\left(\mathbf{r}\right)\psi,\label{eq:bound state problem}
\end{equation}
where
\[
\mathcal{H}_{0}-E\mathcal{I}=\left(\begin{array}{cc}
\sigma_{0}mc^{2} & \frac{\hbar c}{i}\left(\sigma_{1}\frac{\partial}{\partial x_{1}}+\sigma_{2}\frac{\partial}{\partial x_{2}}+\sigma_{3}\frac{\partial}{\partial x_{3}}\right)\\
\frac{\hbar c}{i}\left(\sigma_{1}\frac{\partial}{\partial x_{1}}+\sigma_{2}\frac{\partial}{\partial x_{2}}+\sigma_{3}\frac{\partial}{\partial x_{3}}\right) & -\sigma_{0}mc^{2}
\end{array}\right)-\left(\begin{array}{cc}
\sigma_{0} & \mathbf{0}\\
\mathbf{0} & \sigma_{0}
\end{array}\right)E.
\]
Following \cite{BLA-BAB:2013} and computing 
\[
\left(\mathcal{H}_{0}-E\mathcal{I}\right)\left(\mathcal{H}_{0}+E\mathcal{I}\right)=\mathcal{H}_{0}^{2}-E^{2}\mathcal{I},
\]
we have on the off-diagonal of $\mathcal{H}_{0}^{2}$ 
\[
\sigma_{0}mc^{2}\frac{\hbar c}{i}\left(\sigma_{1}\frac{\partial}{\partial x_{1}}+\sigma_{2}\frac{\partial}{\partial x_{2}}+\sigma_{3}\frac{\partial}{\partial x_{3}}\right)-\frac{\hbar c}{i}\left(\sigma_{1}\frac{\partial}{\partial x_{1}}+\sigma_{2}\frac{\partial}{\partial x_{2}}+\sigma_{3}\frac{\partial}{\partial x_{3}}\right)\sigma_{0}mc^{2}=\mathbf{0}
\]
and on the diagonal 
\[
\sigma_{0}m^{2}c^{4}-\hbar^{2}c^{2}\left(\sigma_{1}\frac{\partial}{\partial x_{1}}+\sigma_{2}\frac{\partial}{\partial x_{2}}+\sigma_{3}\frac{\partial}{\partial x_{3}}\right)^{2}=\sigma_{0}\left(m^{2}c^{4}-\hbar^{2}c^{2}\Delta\right).
\]
Therefore, we have
\[
\mathcal{H}_{0}^{2}-E^{2}\mathcal{I}=\mathcal{I}\left(m^{2}c^{4}-\hbar^{2}c^{2}\Delta-E^{2}\right)
\]
and
\[
\left(\mathcal{H}_{0}+E\mathcal{I}\right)^{-1}\left(\mathcal{H}_{0}-E\mathcal{I}\right)^{-1}=\left(\mathcal{H}_{0}^{2}-E^{2}\mathcal{I}\right)^{-1}=\mathcal{I}\left(-\hbar^{2}c^{2}\Delta+m^{2}c^{4}-E^{2}\right)^{-1}
\]
so that
\[
\left(\mathcal{H}_{0}-E\mathcal{I}\right)^{-1}=\left(\mathcal{H}_{0}+E\mathcal{I}\right)\mathcal{I}\left(-\hbar^{2}c^{2}\Delta+m^{2}c^{4}-E^{2}\right)^{-1}.
\]
As a result, we obtain from (\ref{eq:bound state problem})
\[
\psi=-\frac{1}{\hbar^{2}c^{2}}\left(\mathcal{H}_{0}+E\mathcal{I}\right)\mathcal{I}\left(-\Delta+\frac{m^{2}c^{4}-E^{2}}{c^{2}\hbar^{2}}\right)^{-1}\mathcal{V}\psi.
\]
Noting that for bound states $E<mc^{2}$, we set 
\begin{equation}
\kappa=\frac{\sqrt{m^{2}c^{4}-E^{2}}}{c\hbar},\,\,\,\kappa>0,\label{eq:kappa via E}
\end{equation}
and
\begin{equation}
\left(-\Delta+\frac{m^{2}c^{4}-E^{2}}{c^{2}\hbar^{2}}\right)^{-1}\left(\mathbf{r}-\mathbf{r}'\right)=G\left(\kappa,\mathbf{r}-\mathbf{r}'\right)=\frac{1}{4\pi}\frac{e^{-\kappa\left\Vert \mathbf{r}-\mathbf{r}'\right\Vert }}{\left\Vert \mathbf{r}-\mathbf{r}'\right\Vert }.\label{eq:bstateHelmholtz-GF}
\end{equation}
where the Green's function (\ref{eq:bstateHelmholtz-GF}) solves 
\begin{equation}
\left(-\Delta+\kappa^{2}\right)G\left(\kappa,\mathbf{r}-\mathbf{r}'\right)=\delta\left(\mathbf{r}-\mathbf{r}'\right).\label{eq:BstateHelmholtz}
\end{equation}
We consider an auxiliary eigenvalue problem

\begin{equation}
\lambda\left(\kappa\right)\psi\left(\kappa\right)=-\frac{1}{\hbar^{2}c^{2}}\left(\mathcal{H}_{0}+E\left(\kappa\right)\mathcal{I}\right)\mathcal{G}\left(\kappa\right)\mathcal{V}\psi\left(\kappa\right),\label{eq:auxiliary equation for psi}
\end{equation}
where
\begin{equation}
\mathcal{G}\left(\kappa,\mathbf{r}-\mathbf{r}'\right)=\mathcal{I}G\left(\kappa,\mathbf{r}-\mathbf{r}'\right).\label{eq:Green's function rel case}
\end{equation}
Note that if $\kappa$ is as in (\ref{eq:kappa via E}), where $E$
is an eigenvalue of (\ref{eq:eigenvalue problem Dirac's eq}), then
$\lambda\left(\kappa\right)=1$ and the solution of (\ref{eq:auxiliary equation for psi})
$\psi\left(\kappa\right)$ also solves (\ref{eq:eigenvalue problem Dirac's eq}).
It is convenient to introduce a four component function $\varphi\left(\kappa\right)=\mathcal{G}^{-1/2}\left(\kappa\right)\psi\left(\kappa\right)$
to modify (\ref{eq:auxiliary equation for psi}) so that
\begin{equation}
\lambda\left(\kappa\right)\varphi\left(\kappa\right)=\mathcal{A}\left(\kappa\right)\varphi\left(\kappa\right),\label{eq:auxiliary eq for phi}
\end{equation}
where
\begin{equation}
\mathcal{A}\left(\kappa\right)=-\frac{1}{\hbar^{2}c^{2}}\left(\mathcal{H}_{0}+E\left(\kappa\right)\mathcal{I}\right)\mathcal{G}^{1/2}\left(\kappa\right)\mathcal{V}\mathcal{G}^{1/2}\left(\kappa\right).\label{eq:the main operator}
\end{equation}
We show further below that, in the appropriately chosen Hilbert spaces,
the operator 
\[
-\frac{1}{\hbar^{2}c^{2}}\left(\mathcal{H}_{0}+E\mathcal{I}\right)\mathcal{G}^{1/2}
\]
is bounded, the operator 
\[
\mathcal{V}\mathcal{G}^{1/2}\left(\kappa\right)
\]
is compact, and the spectrum of the operator $\mathcal{A}$ is real.
\begin{rem}
\label{rem:negative definite, relativistic}As in the non-relativistic
case (see Remark~(\ref{rem: neg.definite, non-relativistic})), we
want the matrix potential $\mathcal{V}$ in (\ref{eq:relativistic potential})
to be negative definite. Assuming that $V_{c}\left(\mathbf{r}\right)$
and $V_{xc}\left(\mathbf{r}\right)$ are bounded and using an appropriate
shift $\tau$ of the spectrum in our derivation of (\ref{eq:auxiliary equation for psi})
and (\ref{eq:auxiliary eq for phi}), we write 
\begin{equation}
\left(\mathcal{H}_{0}+\mathcal{V}\left(\mathbf{r}\right)-\tau\mathcal{I}\right)\psi=\left(E-\tau\right)\psi,\label{eq:eigenvalue problem Dirac's  shifted}
\end{equation}
where $\tau$ is sufficiently large. In such case $E-\tau<E<mc^{2}$
and we set 
\begin{equation}
\kappa=\frac{\sqrt{m^{2}c^{4}-\left(E-\tau\right)^{2}}}{c\hbar},\,\,\,\kappa>0.\label{eq:kappa via E-1}
\end{equation}
We then have 
\begin{equation}
\mathcal{A}\left(\kappa\right)=-\frac{1}{\hbar^{2}c^{2}}\left(\mathcal{H}_{0}+E\left(\kappa\right)\mathcal{I}\right)\mathcal{G}^{1/2}\left(\kappa\right)\left(\mathcal{V}-\tau\mathcal{I}\right)\mathcal{G}^{1/2}\left(\kappa\right).\label{eq:the main operator-2}
\end{equation}
As in the non-relativistic case, components of $\varphi$ in (\ref{eq:auxiliary eq for phi})
are computationally considered to be non-zero in a bounded domain.
Again, so far in our computations we did not encounter a need to shift
the spectrum. The shift of the spectrum (if it were needed) allows
us to consider $\mathcal{V}$ to be a negative definite matrix multiplication
operator without loss of generality.
\end{rem}

\section{Hilbert spaces for solutions of equations of quantum chemistry}

For our analysis it is convenient to consider equations of quantum
chemistry in momentum space, where the Green's function component
in (\ref{eq:Green's function rel case}) is
\[
G\left(\kappa,\left\Vert \mathbf{p}\right\Vert \right)=\left(\kappa^{2}+\left\Vert \mathbf{p}\right\Vert ^{2}\right)^{-1}.
\]
If $\widehat{\psi}\left(\mathbf{p}\right)$ is a component of the
spinor solution of the Dirac equation, we require

\[
\widehat{\varphi}\left(\mathbf{p}\right)=\left(\kappa^{2}+\left\Vert \mathbf{p}\right\Vert ^{2}\right)^{1/2+\delta/2}\widehat{\psi}\left(\mathbf{p}\right)\in L^{2}\left(\mathbb{R}^{3}\right),
\]
where $\kappa>0$ and $\delta>0$, i.e. the function $\psi$ belongs
to the Hilbert space $\mathscr{H}_{\kappa,\delta}$ with the weighted
inner product 
\begin{equation}
\left\langle \psi_{1},\psi_{2}\right\rangle _{\kappa,\delta}=\int_{\mathbb{R}^{3}}\left(\kappa^{2}+\left\Vert \mathbf{p}\right\Vert ^{2}\right)^{1+\delta}\widehat{\psi_{1}}\left(\mathbf{p}\right)\overline{\widehat{\psi_{2}}\left(\mathbf{p}\right)}d\mathbf{p}\label{eq:weighted inner product}
\end{equation}
and the corresponding norm,
\begin{equation}
\left\Vert \psi\right\Vert _{\kappa,\delta}=\left(\int_{\mathbb{R}^{3}}\left(\kappa^{2}+\left\Vert \mathbf{p}\right\Vert ^{2}\right)^{1+\delta}\left|\widehat{\psi}\left(\mathbf{p}\right)\right|^{2}d\mathbf{p}\right)^{1/2}.\label{eq:weighed norm}
\end{equation}
For non-relativistic equations this condition (for $\widehat{\psi_{\mu}}\left(\mathbf{p}\right)$
in (\ref{eq:auxiliary eigenvalue problem Kohn-Sham})) is easily satisfied
since the worst singularity of a solution is a cusp at the location
of a nuclei, e.g. $e^{-\left\Vert \mathbf{r}\right\Vert }$, which
in momentum space corresponds to $\left(1+\left\Vert \mathbf{p}\right\Vert ^{2}\right)^{-2}$.
Consequently, for large $\left\Vert \mathbf{p}\right\Vert $, the
asymptotic rate of decay of non-relativistic bound states in momentum
space is $\left\Vert \mathbf{p}\right\Vert ^{-4}$ which is sufficient
to keep the integral in (\ref{eq:weighed norm}) finite. Note that
the eigenfunctions with a polynomial factor that is zero at the origin
decay even faster in momentum space.

The solutions of Dirac's equation have a stronger singularity at the
location of a nuclei, e.g.
\begin{equation}
\left\Vert \mathbf{r}\right\Vert ^{\gamma\left(Z\right)-1}e^{-\left\Vert \mathbf{r}\right\Vert },\label{eq:Dirac singularity in the momentum space}
\end{equation}
with
\[
\gamma\left(Z\right)=\left(1-\left(\frac{Z}{c}\right)^{2}\right)^{1/2},\,\,\,\,0<\gamma\left(Z\right)<1,
\]
where $Z$ is the charge of the nucleus and $c$ is the speed of light,
$c\approx137.035999084$ in atomic units (see e.g. \cite[Section 2.3]{KUTZEL:1989},
\cite[Section 3.1]{KUTZEL:1989a}). Note that to estimate the decay
we cannot use a stronger singularity $\left\Vert \mathbf{r}\right\Vert ^{-1}e^{-\left\Vert \mathbf{r}\right\Vert }$
instead of (\ref{eq:Dirac singularity in the momentum space}) since,
in momentum space, it corresponds to $\left(1+\left\Vert \mathbf{p}\right\Vert ^{2}\right)^{-1}$and
this rate of decay is too slow to keep (\ref{eq:weighed norm}) finite.
Therefore, we need to estimate the rate of decay of solutions of Dirac's
equations in momentum space for (\ref{eq:Dirac singularity in the momentum space})
directly. Computing the Fourier transform of (\ref{eq:Dirac singularity in the momentum space}),
we obtain (see \cite[Eq.3.381.5]{GRA-RYZ:2007})
\begin{eqnarray*}
\psi_{0}\left(p\right) & = & \int_{\mathbb{R}^{3}}\left\Vert \mathbf{r}\right\Vert ^{\gamma\left(Z\right)-1}e^{-\left\Vert \mathbf{r}\right\Vert }e^{-i\mathbf{r}\cdot\mathbf{p}}d\mathbf{r}\\
 & = & \frac{4\pi}{p}\int_{0}^{\infty}e^{-r}r^{\gamma\left(Z\right)}\sin\left(pr\right)dr\\
 & = & \frac{4\pi}{p}\frac{\Gamma\left(1+\gamma\left(Z\right)\right)\sin\left[\left(1+\gamma\left(Z\right)\right)\arctan\left(p\right)\right]}{\left(1+p^{2}\right)^{1/2+\gamma\left(Z\right)/2}},
\end{eqnarray*}
where $r=\left\Vert \mathbf{r}\right\Vert $ and $p=\left\Vert \mathbf{p}\right\Vert $.
Since $\gamma\left(Z\right)<1$ and for large $p$
\[
\sin\left[\left(1+\gamma\left(Z\right)\right)\arctan\left(p\right)\right]=\sin\left[\frac{1}{2}\left(1+\gamma\left(Z\right)\right)\pi\right]-\frac{1}{p}\left(1+\gamma\left(Z\right)\right)\cos\left[\frac{1}{2}\left(1+\gamma\left(Z\right)\right)\pi\right]+\mathcal{O}\left(\frac{1}{p^{2}}\right),
\]
we obtain 
\[
\psi_{0}\left(p\right)=\frac{4\pi}{p}\frac{\Gamma\left(1+\gamma\left(Z\right)\right)\sin\left[\frac{1}{2}\left(1+\gamma\left(Z\right)\right)\pi\right]}{\left(1+p^{2}\right)^{1/2+\gamma\left(Z\right)/2}}+\mathcal{O}\left(\frac{1}{p^{3+\gamma\left(Z\right)}}\right).
\]
Estimating the norm (\ref{eq:weighed norm}), we observe that the
integrand in (\ref{eq:weighed norm}) behaves as 
\[
\frac{1}{p^{2+2\gamma\left(Z\right)-2\delta}}
\]
for large $p$. For convergence we need $2+2\gamma\left(Z\right)-2\delta>3$
or

\[
\gamma\left(Z\right)>\frac{1}{2}+\delta.
\]
Since $\gamma\left(Z\right)$ is a monotone function and $\gamma\left(118\right)=0.508457$,
we conclude that, for a sufficiently small $\delta>0$, we can use
the Hilbert space $\mathscr{H}_{\kappa,\delta}$ with the weighted
inner product (\ref{eq:weighted inner product}) as a space for the
bound states of the Dirac's equations for nuclei with charges $1\le Z\le118$.

Our interest in considering the Hilbert space $\mathscr{H}_{\kappa,\delta}$
is more theoretical than practical. In all practical computations
the singularity of the nuclear potentials is removed either explicitly
or implicitly as a result of either using a finite computational basis
or grid or through the use of a more physical finite charge distribution
of the nucleus. As we discuss next, if we consider solutions in $\mathscr{H}_{\kappa,0}$
(i.e. set $\delta=0$ in (\ref{eq:weighted inner product})), then
the operators we construct are compact for an \textit{arbitrarily
accurate} \textit{approximation} of the Coulomb potential for any
nuclei charge $Z<c$. It turns out that by considering solutions in
$\mathscr{H}_{\kappa,\delta}$, for nuclei charges $1\le Z\le118$
the operators in question are compact for the Coulomb potential itself,
without any approximation. In any case, the practical impact of our
considerations is that the spectrum of the family of operators of
the auxiliary eigenvalue problems (\ref{eq:auxiliary eigenvalue problem Kohn-Sham-1})
and (\ref{eq:auxiliary eq for phi}) is always discrete and bounded
from below which, in turn, assures convergence of an iterative approach
for computing the bound states.

\section{The Hilbert-Schmidt operators}

We start by considering the matrix operator $\mathcal{V}\mathcal{G}^{1/2}$
which, in the momentum space, can be written as 
\begin{equation}
\left[\mathcal{V}\mathcal{G}^{1/2}\right]\left(\kappa,\mathbf{p},\mathbf{p}'\right)=\mathcal{I}\frac{\widehat{V}\left(\mathbf{p}-\mathbf{p}'\right)}{\left(\kappa^{2}+\left\Vert \mathbf{p}'\right\Vert ^{2}\right)^{1/2}},\label{eq:relativistic combination delta=00003D0}
\end{equation}
acting on functions 
\[
\widehat{\varphi}\left(\mathbf{p}\right)=\left(\kappa^{2}+\left\Vert \mathbf{p}\right\Vert ^{2}\right)^{1/2}\widehat{\psi}\left(\mathbf{p}\right)\in L^{2}\left(\mathbb{R}^{3}\right),
\]
or $\psi\in\mathscr{H}_{\kappa,0}.$ Alternatively, we can consider
the matrix operator
\begin{equation}
\left[\mathcal{V}\mathcal{G}^{1/2}\right]_{\delta}\left(\kappa,\mathbf{p},\mathbf{p}'\right)=\mathcal{I}\frac{\widehat{V}\left(\mathbf{p}-\mathbf{p}'\right)}{\left(\kappa^{2}+\left\Vert \mathbf{p}'\right\Vert ^{2}\right)^{1/2+\delta/2}},\label{eq:relativistic combination delta > 0}
\end{equation}
acting on functions 
\[
\widehat{\varphi}\left(\mathbf{p}\right)=\left(\kappa^{2}+\left\Vert \mathbf{p}\right\Vert ^{2}\right)^{1/2+\delta/2}\widehat{\psi}\left(\mathbf{p}\right)\in L^{2}\left(\mathbb{R}^{3}\right),
\]
or $\psi\in\mathscr{H}_{\kappa,\delta}$, a class of functions $\psi$
decaying slightly faster in the momentum space.

Combining these operators with their Hermitian adjoints which we denote
by $\,^{*},$ we obtain

\[
\left[\mathcal{V}\mathcal{G}^{1/2}\right]^{*}\left[\mathcal{V}\mathcal{G}^{1/2}\right]=\mathcal{G}^{1/2}\mathcal{V}^{2}\mathcal{G}^{1/2}
\]
and 
\[
\left[\mathcal{V}\mathcal{G}^{1/2}\right]_{\delta}^{*}\left[\mathcal{V}\mathcal{G}^{1/2}\right]_{\delta}=\mathcal{G}_{\delta}^{1/2}\mathcal{V}^{2}\mathcal{G}_{\delta}^{1/2}.
\]
Our goal is to show that the matrix operators $\mathcal{V}\mathcal{G}^{1/2}$
and $\left[\mathcal{V}\mathcal{G}^{1/2}\right]_{\delta}$ are compact;
for this we rely on Lemma~\ref{thm:A*A} (see Appendix~\ref{subsec:Appendix-A:-Compact A^*A})
showing that the compactness of operators $\mathcal{G}^{1/2}\mathcal{V}^{2}\mathcal{G}^{1/2}$
(or $\mathcal{G}_{\delta}^{1/2}\mathcal{V}^{2}\mathcal{G}_{\delta}^{1/2}$)
implies compactness of $\left[\mathcal{V}\mathcal{G}^{1/2}\right]$
(or $\left[\mathcal{V}\mathcal{G}^{1/2}\right]_{\delta}$), respectively.

We study components of these matrix operators
\[
G^{1/2}V^{2}G^{1/2},
\]
and 
\[
G_{\delta}^{1/2}V^{2}G_{\delta}^{1/2},
\]
which have the kernels
\begin{equation}
K\left(\kappa,\mathbf{p},\mathbf{p}'\right)=\frac{\widehat{V^{2}}\left(\mathbf{p}-\mathbf{p}'\right)}{\left(\kappa^{2}+\left\Vert \mathbf{p}\right\Vert ^{2}\right)^{1/2}\left(\kappa^{2}+\left\Vert \mathbf{p}'\right\Vert ^{2}\right)^{1/2}}\label{eq:kernel (with delta=00003D0)}
\end{equation}
and
\begin{equation}
K_{\delta}\left(\kappa,\mathbf{p},\mathbf{p}'\right)=\frac{\widehat{V^{2}}\left(\mathbf{p}-\mathbf{p}'\right)}{\left(\kappa^{2}+\left\Vert \mathbf{p}\right\Vert ^{2}\right)^{1/2+\delta/2}\left(\kappa^{2}+\left\Vert \mathbf{p}'\right\Vert ^{2}\right)^{1/2+\delta/2}}.\label{eq:kernel delta}
\end{equation}

\subsection{The Rollnik class of potentials}

We have
\begin{thm}
\label{thm:()-If-potential V^2}(\cite[Theorem I.22]{SIMON:1971})
If potential $V^{2}$ is in the Rollnik class, then the operator $G^{1/2}V^{2}G^{1/2}$
with the kernel \textbf{$K\left(\kappa,\mathbf{p},\mathbf{p}'\right)$}
is a bounded Hilbert-Schmidt operator.
\end{thm}

In other words, if the potential $V^{2}$ is in the Rollnik class,
then the integral 
\begin{equation}
\left\Vert K\right\Vert _{HS}=\left(\int_{\mathbb{R}^{3}}\int_{\mathbb{R}^{3}}\frac{\left|\widehat{V^{2}}\left(\mathbf{p}-\mathbf{p}'\right)\right|^{2}}{\left(\kappa^{2}+\left\Vert \mathbf{p}\right\Vert ^{2}\right)\left(\kappa^{2}+\left\Vert \mathbf{p}'\right\Vert ^{2}\right)}d\mathbf{p}'d\mathbf{p}\right)^{1/2}\label{eq:Hilbert-Schmidt norm}
\end{equation}
is finite. It is well known that the Hilbert-Schmidt operators are
compact.

The Rollnik class of potentials is defined by the condition (see \cite{SIMON:1971})
\[
\int_{\mathbb{R}^{3}\times\mathbb{R}^{3}}\frac{\left|V\left(\mathbf{r}\right)\right|\left|V\left(\mathbf{r}'\right)\right|}{\left\Vert \mathbf{r}-\mathbf{r}'\right\Vert ^{2}}d\mathbf{r}d\mathbf{r}'<\infty,
\]
which can be shown to be equivalent to (\ref{eq:Hilbert-Schmidt norm})
(see \cite[Theorem I.22]{SIMON:1971}). Note that in our case it is
a requirement for $V^{2}$ rather than $V$. The Rollnik class of
potentials has been identified independently by a number of authors
\cite{ROLLNI:1956,SCHWIN:1961,GRO-WU:1961,GRO-WU:1962,SC-WE-WR:1964}
as a class of physically significant potentials that have workable
mathematical properties. Quoting from \cite[page 126]{SCHWIN:1961}:
``The latter quantity exists for potentials that decrease more rapidly
than $\left\Vert \mathbf{r}\right\Vert ^{-2}$ as $\left\Vert \mathbf{r}\right\Vert \to\infty$
and that in the neighborhoods of a finite number of points $\mathbf{r}_{0}$
are less singular than $\left\Vert \mathbf{r}-\mathbf{r}_{0}\right\Vert ^{-2}$''.
Note that since we consider the square of the potential in Theorem~\ref{thm:()-If-potential V^2},
for the Coulomb potential we are just at the threshold of missing
the Rollnik class. 

In practical computation we replace the Coulomb potential by a linear
combination of Gaussians. Following \cite[Theorem 3, 5 and Lemma 4]{BEY-MON:2010}
and setting 
\begin{equation}
S_{\infty}\left(\left\Vert \mathbf{r}\right\Vert \right)=\frac{h}{\Gamma(\alpha/2)}\sum_{n\in\mathbb{Z}}e^{\alpha nh/2}e^{-e^{nh}\left\Vert \mathbf{r}\right\Vert ^{2}},\label{S-infinity}
\end{equation}
we have 
\begin{thm}
Given $\alpha>0$ and $0<\epsilon\leq1$, for any step size $h$ such
that 
\begin{equation}
h\leq\frac{2\pi}{\log3+\alpha\log(\cos1)^{-1}/2+\log\epsilon^{-1}},\label{hEstimate}
\end{equation}
we have\textup{ }for $\left\Vert \mathbf{r}\right\Vert >0$
\begin{equation}
\left|\left\Vert \mathbf{r}\right\Vert ^{-\alpha}-S_{\infty}\left(\left\Vert \mathbf{r}\right\Vert \right)\right|\leq\left\Vert \mathbf{r}\right\Vert ^{-\alpha}\epsilon,\label{seriesApproximation}
\end{equation}
and
\[
S_{\infty}\left(\left\Vert \mathbf{r}\right\Vert \right)<(\epsilon+1)\left\Vert \mathbf{r}\right\Vert ^{-\alpha},
\]
where $S_{\infty}$ is given in (\ref{S-infinity}).
\end{thm}

For a given accuracy $\epsilon$ and power $\alpha$ (n.b., for the
Coulomb potential $\alpha=1$ and for the square $\alpha=2$), we
may first select $h$ and then, for a given range of values $\left\Vert \mathbf{r}\right\Vert $,
truncate $S_{\infty}\left(\left\Vert \mathbf{r}\right\Vert \right)$
to yield a finite sum approximation in that range to obtain a finite
sum $S_{F}\left(\left\Vert \mathbf{r}\right\Vert \right)$, 
\begin{equation}
S_{F}\left(\left\Vert \mathbf{r}\right\Vert \right)=S_{F}\left(\left\Vert \mathbf{r}\right\Vert ;M_{0},M_{1},h\right)=\sum_{n=M_{0}+1}^{M_{1}}e^{\alpha nh/2}e^{-e^{nh}\left\Vert \mathbf{r}\right\Vert ^{2}}.\label{finiteApproxDef}
\end{equation}
It is shown in \cite[Theorem 5]{BEY-MON:2010} that, for a fixed $\alpha$,
$d<\left\Vert \mathbf{r}\right\Vert <1/d$ and any finite $\epsilon>0$,
the step $h=\mathcal{O}\left(1/\log\epsilon^{-1}\right)$ and the
number of terms in (\ref{finiteApproxDef}) are estimated as $M_{1}-M_{0}=\mathcal{O}\left(\log d^{-1}\right)$
and $M_{1}-M_{0}=\mathcal{O}\left(\left(\log\epsilon^{-1}\right)^{2}\right)$.
By choosing $d$ to be small (e.g. $d=10^{-15}$) and selecting $\epsilon$
as needed, we replace the Coulomb potentials by their approximation
for any user-selected accuracy and range. Importantly, the resulting
approximating potential $S_{F}\left(\left\Vert \mathbf{r}\right\Vert \right)$
is in the Rollnik class so that the kernel (\ref{eq:kernel (with delta=00003D0)})
is that of a compact operator.

\subsection{A compact operator for the Coulomb potential}

As has already been mentioned, neither the Coulomb potential nor its
square are in the Rollnik class. However, we show that for the Coulomb
potential the operator with the kernel (\ref{eq:kernel delta}) is
compact for functions in the Hilbert space $\mathscr{H}_{\kappa,\delta}$
for any $\delta>0$. Using (\ref{eq:kernel delta}) with the Coulomb
potential, 
\[
\widehat{V^{2}}\left(\mathbf{p}-\mathbf{p}'\right)=\frac{1}{\left\Vert \mathbf{p}-\mathbf{p}'\right\Vert },
\]
we consider the kernel

\begin{equation}
T\left(\mathbf{p},\mathbf{p}'\right)=\frac{1}{\left(\kappa^{2}+\left\Vert \mathbf{p}\right\Vert ^{2}\right)^{1/2+\delta/2}}\frac{1}{\left\Vert \mathbf{p}-\mathbf{p}'\right\Vert }\frac{1}{\left(\kappa^{2}+\left\Vert \mathbf{p}'\right\Vert ^{2}\right)^{1/2+\delta/2}}.\label{eq:operator Tdelta}
\end{equation}

\begin{lem}
\label{lem:The-operator-with}The operator with the kernel (\ref{eq:operator Tdelta})
is a Hilbert-Schmidt operator with the Hilbert-Schmidt norm 
\[
\left\Vert T\right\Vert _{HS}=\frac{1}{\kappa^{2\delta}}\frac{\Gamma\left(1/2+\delta\right)}{\Gamma\left(1+\delta\right)}\frac{\pi^{3/3}}{\sqrt{\delta}}.
\]
\end{lem}

\begin{proof}
We prove Lemma~\ref{lem:The-operator-with} by explicitly evaluating
the Hilbert-Schmidt norm by computing the integral 
\begin{equation}
\left\Vert T\right\Vert _{HS}^{2}=\int_{\mathbb{R}^{3}}\int_{\mathbb{R}^{3}}\frac{1}{\left(\kappa^{2}+\left\Vert \mathbf{p}\right\Vert ^{2}\right)^{1+\delta}}\frac{1}{\left\Vert \mathbf{p}-\mathbf{p}'\right\Vert ^{2}}\frac{1}{\left(\kappa^{2}+\left\Vert \mathbf{p}'\right\Vert ^{2}\right)^{1+\delta}}d\mathbf{p}'d\mathbf{p}.\label{eq:integral for HS norm}
\end{equation}
We start with two integrals,
\begin{equation}
\frac{1}{\left\Vert \mathbf{p}\right\Vert ^{2}}=\int_{0}^{\infty}e^{-\tau\left\Vert \mathbf{p}\right\Vert ^{2}}d\tau,\label{eq:one over psquare}
\end{equation}
and
\begin{equation}
\frac{1}{\left(\kappa^{2}+\left\Vert \mathbf{p}\right\Vert ^{2}\right)^{1+\delta}}=\frac{1}{\Gamma\left(1+\delta\right)}\int_{0}^{\infty}e^{-s\left(\kappa^{2}+\left\Vert \mathbf{p}\right\Vert ^{2}\right)}s^{\delta}ds.\label{eq:one over with kappa}
\end{equation}
Substituting (\ref{eq:one over psquare}) and (\ref{eq:one over with kappa})
into (\ref{eq:integral for HS norm}), we have 
\begin{eqnarray*}
\left\Vert T\right\Vert _{HS}^{2} & = & \frac{1}{\Gamma\left(1+\delta\right)^{2}}\int_{0}^{\infty}\int_{0}^{\infty}\int_{0}^{\infty}dsdtd\tau\,\left(st\right)^{\delta}e^{-\kappa^{2}\left(s+t\right)}\\
 &  & \int_{\mathbb{R}^{3}}\int_{\mathbb{R}^{3}}e^{-s\left\Vert \mathbf{p}\right\Vert ^{2}}e^{-s\left\Vert \mathbf{p}'\right\Vert ^{2}}e^{-\tau\left\Vert \mathbf{p}-\mathbf{p}'\right\Vert ^{2}}d\mathbf{p}'d\mathbf{p}
\end{eqnarray*}
Evaluating the integral with the Gaussians, we obtain 
\[
\int_{\mathbb{R}^{3}}\int_{\mathbb{R}^{3}}e^{-s\left\Vert \mathbf{p}\right\Vert ^{2}}e^{-s\left\Vert \mathbf{p}'\right\Vert ^{2}}e^{-\tau\left\Vert \mathbf{p}-\mathbf{p}'\right\Vert ^{2}}d\mathbf{p}'d\mathbf{p}=\frac{\pi^{3}}{\left(\tau\left(s+t\right)+ts\right)^{3/2}}.
\]
Next we compute the integral over $\tau$ and obtain

\[
\int_{0}^{\infty}\frac{\pi^{3}}{\left(\tau\left(s+t\right)+ts\right)^{3/2}}d\tau=\frac{2\pi^{3}}{\left(s+t\right)\left(st\right)^{1/2}}.
\]
Finally, we compute 
\[
\left\Vert T\right\Vert _{HS}^{2}=\frac{2\pi^{3}}{\Gamma\left(1+\delta\right)^{2}}\int_{0}^{\infty}\int_{0}^{\infty}e^{-\kappa^{2}\left(s+t\right)}\frac{1}{\left(s+t\right)\left(st\right)^{1/2-\delta}}dsdt=\frac{\pi^{3}\Gamma\left(1/2+\delta\right)^{2}}{\delta\,\Gamma\left(1+\delta\right)^{2}}\frac{1}{\kappa^{4\delta}}.
\]
\end{proof}
\begin{rem}
A product of two Coulomb potentials with distinct nuclear centers
has the same asymptotic decay in space as the square of the Coulomb
potential considered above. Therefore, the product in the momentum
space will have the same dominant singularity at the origin as the
square of the Coulomb potential; other components of the potential
are sufficiently smooth in space so that their decay is faster in
the momentum space and will not cause the Hilbert-Schmidt norm to
become unbounded.
\end{rem}

\section{Bounded operators}

Next we consider two operators of interest and show that they are
bounded.

\subsection{Operator \textmd{$G_{\mu}^{1/2}$}}

In the momentum space the operator $G_{\mu}^{1/2}$ (\ref{eq:Green's function, non-rel case})
is a multiplication operator by 
\[
\frac{1}{\left(\mu^{2}+\left\Vert \mathbf{p}\right\Vert ^{2}\right)^{1/2}}
\]
and, therefore, is a bounded operator with the norm $1/\mu$.

\subsection{Matrix operator $\frac{1}{\hbar^{2}c^{2}}\left(\mathcal{H}_{0}+E\mathcal{I}\right)\mathcal{G}^{1/2}$}

In momentum space we have
\[
\widehat{\frac{1}{\hbar^{2}c^{2}}\left(\mathcal{H}_{0}+E\mathcal{I}\right)}=\left(\begin{array}{cccc}
\frac{m}{\hbar^{2}}+\frac{E}{\hbar^{2}c^{2}} & 0 & \frac{p_{3}}{\hbar c} & \frac{p_{1}-ip_{2}}{\hbar c}\\
0 & \frac{m}{\hbar^{2}}+\frac{E}{\hbar^{2}c^{2}} & \frac{p_{1}+ip_{2}}{\hbar c} & \frac{-p_{3}}{\hbar c}\\
\frac{p_{3}}{\hbar c} & \frac{p_{1}-ip_{2}}{\hbar c} & -\frac{m}{\hbar^{2}}+\frac{E}{\hbar^{2}c^{2}} & 0\\
\frac{p_{1}+ip_{2}}{\hbar c} & \frac{-p_{3}}{\hbar c} & 0 & -\frac{m}{\hbar^{2}}+\frac{E}{\hbar^{2}c^{2}}
\end{array}\right)
\]
and
\[
\mathcal{G}^{1/2}\left(\kappa,\mathbf{p}\right)=\mathcal{I}\frac{1}{\left(\kappa^{2}+\left\Vert \mathbf{p}\right\Vert ^{2}\right)^{1/2}}.
\]
We need to show that 
\[
\left\Vert \frac{1}{\hbar^{2}c^{2}}\left(\mathcal{H}_{0}+E\mathcal{I}\right)\mathcal{G}^{1/2}\varphi\right\Vert _{2}^{2}\le Const\left\Vert \varphi\right\Vert _{2}^{2},
\]
where the norm is the sum of squares of absolute values of the components
and the integration in $\mathbb{R}^{3}$ is over the variable $\mathbf{p}$.
For the proof we split the diagonal and off-diagonal parts of the
matrix operator,
\[
\widehat{\frac{1}{\hbar^{2}c^{2}}\left(\mathcal{H}_{0}+E\mathcal{I}\right)}\mathcal{G}^{1/2}=\mathcal{D}+\mathcal{O},
\]
where
\[
\mathcal{D}=\left(\kappa^{2}+\left\Vert \mathbf{p}\right\Vert ^{2}\right)^{-1/2}\left(\begin{array}{cccc}
\frac{m}{\hbar^{2}}+\frac{E}{\hbar^{2}c^{2}} & 0 & 0 & 0\\
0 & \frac{m}{\hbar^{2}}+\frac{E}{\hbar^{2}c^{2}} & 0 & 0\\
0 & 0 & -\frac{m}{\hbar^{2}}+\frac{E}{\hbar^{2}c^{2}} & 0\\
0 & 0 & 0 & -\frac{m}{\hbar^{2}}+\frac{E}{\hbar^{2}c^{2}}
\end{array}\right)
\]
and
\[
\mathcal{O}=\left(\kappa^{2}+\left\Vert \mathbf{p}\right\Vert ^{2}\right)^{-1/2}\left(\begin{array}{cccc}
0 & 0 & \frac{p_{3}}{\hbar c} & \frac{p_{1}-ip_{2}}{\hbar c}\\
0 & 0 & \frac{p_{1}+ip_{2}}{\hbar c} & \frac{-p_{3}}{\hbar c}\\
\frac{p_{3}}{\hbar c} & \frac{p_{1}-ip_{2}}{\hbar c} & 0 & 0\\
\frac{p_{1}+ip_{2}}{\hbar c} & \frac{-p_{3}}{\hbar c} & 0 & 0
\end{array}\right).
\]
Let $u\left(\mathbf{p}\right)=\left(u_{1}\left(\mathbf{p}\right),u_{2}\left(\mathbf{p}\right),u_{3}\left(\mathbf{p}\right),u_{4}\left(\mathbf{p}\right)\right)^{T}$
and compute 
\[
\int_{\mathbb{R}^{3}}\left\langle \mathcal{O}u,\mathcal{O}u\right\rangle d\mathbf{p}=\frac{1}{\hbar^{2}c^{2}}\int_{\mathbb{R}^{3}}\frac{\left\Vert \mathbf{p}\right\Vert ^{2}}{\kappa^{2}+\left\Vert \mathbf{p}\right\Vert ^{2}}\left\Vert u\left(\mathbf{p}\right)\right\Vert ^{2}d\mathbf{p}\le\frac{1}{\hbar^{2}c^{2}}\int_{\mathbb{R}^{3}}\left\Vert u\left(\mathbf{p}\right)\right\Vert ^{2}d\mathbf{p}
\]
so that we have 
\[
\left\Vert \mathcal{O}\right\Vert _{2}\le\frac{1}{\hbar c}.
\]
For the diagonal part, we have
\begin{eqnarray*}
\int_{\mathbb{R}^{3}}\left\langle \mathcal{D}u,\mathcal{D}u\right\rangle d\mathbf{p} & = & \frac{E}{\hbar^{2}c^{2}}\int_{\mathbb{R}^{3}}\frac{1}{\kappa^{2}+\left\Vert \mathbf{p}\right\Vert ^{2}}\left\Vert u\left(\mathbf{p}\right)\right\Vert ^{2}d\mathbf{p}\\
 & + & \frac{m}{c^{2}}\int_{\mathbb{R}^{3}}\frac{1}{\kappa^{2}+\left\Vert \mathbf{p}\right\Vert ^{2}}\left(\left|u_{1}\left(\mathbf{p}\right)\right|^{2}+\left|u_{2}\left(\mathbf{p}\right)\right|^{2}-\left|u_{3}\left(\mathbf{p}\right)\right|^{2}-\left|u_{4}\left(\mathbf{p}\right)\right|^{2}\right)d\mathbf{p}
\end{eqnarray*}
or
\[
\int_{\mathbb{R}^{3}}\left\langle \mathcal{D}u,\mathcal{D}u\right\rangle d\mathbf{p}\le\left(\frac{E}{\hbar^{2}c^{2}\kappa^{2}}+\frac{m}{c^{2}\kappa^{2}}\right)\int_{\mathbb{R}^{3}}\left\Vert u\left(\mathbf{p}\right)\right\Vert ^{2}d\mathbf{p}
\]
so that 
\[
\left\Vert \mathcal{D}\right\Vert _{2}\le\left(\frac{E}{\hbar^{2}c^{2}\kappa^{2}}+\frac{m}{c^{2}\kappa^{2}}\right)^{1/2}
\]
and the operator in question is bounded.

\section{Spectral structure of auxiliary family of integral operators}

Since the product of a bounded and a compact operator is compact,
we have shown that in the non-relativistic case the operator in (\ref{eq:auxiliary eigenvalue problem Kohn-Sham-1}),
\[
G_{\mu}^{1/2}VG_{\mu}^{1/2},
\]
is compact for any $\mu>0$ and the operator $\mathcal{A}\left(\kappa\right)$
in (\ref{eq:the main operator}) is compact for any $\kappa>0$.

The operator $G_{\mu}^{1/2}VG_{\mu}^{1/2}$ is compact and self-adjoint,
so we can apply the spectral theorem for compact self-adjoint operators
in a Hilbert space. Therefore, we know that it has only discrete eigenvalues
with the only possible accumulation point at zero, and its norm is
equal to the largest absolute value of an eigenvalue. Considering

\begin{equation}
\lambda_{\mu}\phi_{\mu}=-2G_{\mu}^{1/2}VG_{\mu}^{1/2}\phi_{\mu},\label{eq:aux eig prob}
\end{equation}
we have 
\begin{equation}
\lambda_{\mu}\left\langle \phi_{\mu},\phi_{\mu}\right\rangle =-2\left\langle G_{\mu}^{1/2}VG_{\mu}^{1/2}\phi_{\mu},\phi_{\mu}\right\rangle =-2\left\langle VG_{\mu}^{1/2}\phi_{\mu},G_{\mu}^{1/2}\phi_{\mu}\right\rangle =-2\left\langle V\psi_{\mu},\psi_{\mu}\right\rangle .\label{eq:lambda is positive}
\end{equation}

\textit{Since the potential (\ref{eq: components of the potential})
is a negative definite multiplication operator (see Remark~\ref{rem: neg.definite, non-relativistic}),
it implies that $\lambda_{\mu}>0$ for any $\mu>0$.} Note that from
(\ref{eq:Kohn-Sham}) we have 
\[
\left\langle V\left(\mathbf{r}\right)\psi_{i},\psi_{i},\right\rangle =E_{i}\left\langle \psi_{i},\psi_{i}\right\rangle -\frac{1}{2}\left\langle \nabla\psi_{i},\nabla\psi_{i}\right\rangle ,\,\,\,i=1,\dots.N,
\]
and, if energies $E_{i}$ are negative, then on the solutions $\psi_{i}$
\[
\left\langle V\left(\mathbf{r}\right)\psi_{i},\psi_{i}\right\rangle <0.
\]
Each time we have $\mu_{i}^{2}=-E_{i}$, it implies that $\lambda_{\mu}=1$
and $\psi_{\mu_{i}}$ is a bound state satisfying (\ref{eq:Kohn-Sham}).
Also if $\lambda_{\mu}=1$ then $\mu_{i}^{2}=-E_{i}$ so that we can
seek unit eigenvalues of the auxiliary eigenvalue problem as a way
of finding the bound states.

In the relativistic case, we consider the operator $\mathcal{A}\left(\kappa\right)$,
\begin{equation}
\mathcal{A}\left(\kappa\right)=-\frac{1}{\hbar^{2}c^{2}}\left(\mathcal{H}_{0}+E\left(\kappa\right)\mathcal{I}\right)\mathcal{G}^{1/2}\left(\kappa\right)\mathcal{V}\mathcal{G}^{1/2}\left(\kappa\right),\label{eq:the main operator-1}
\end{equation}
which is a product of two self-adjoint operators. \textit{Since the
potential (\ref{eq:relativistic potential}) is a negative definite
matrix multiplication operator (see Remark~\ref{rem:negative definite, relativistic}),}
the spectrum of $\mathcal{A}\left(\kappa\right)$ is real. Indeed,
let us define the operator 
\[
\mathcal{B}\left(\kappa\right)=\mathcal{G}^{1/2}\left(\kappa\right)\left(\mathcal{-V}\right)\mathcal{G}^{1/2}\left(\kappa\right),
\]
which is a compact positive-definite operator. Since the operator
$\mathcal{B}\left(\kappa\right)$ is positive definite and self-adjoint,
according to corollary of Lemma~\ref{thm:A*A}, $\mathcal{B}^{1/2}\left(\kappa\right)$
is well defined and is also a compact operator. Since $B$ is positive
definite, $\mathcal{B}^{-1/2}$ is well defined and is at least bounded.
Since $\mathcal{A}\left(\kappa\right)$ is a compact operator, we
can consider
\[
\mathcal{B}^{1/2}\left(\kappa\right)\mathcal{A}\left(\kappa\right)\mathcal{B}^{-1/2}=\mathcal{B}^{1/2}\left(\kappa\right)\left[\frac{1}{\hbar^{2}c^{2}}\left(\mathcal{H}_{0}+E\left(\kappa\right)\mathcal{I}\right)\right]\mathcal{B}^{1/2}\left(\kappa\right)
\]
which is similar to $\mathcal{A}\left(\kappa\right)$ and is compact
and self-adjoint. Therefore, we conclude that the spectrum of $\mathcal{A}\left(\kappa\right)$
is real and has only discrete eigenvalues with the only possible accumulation
point at zero. Note that our considerations above are valid for any
$\mu>0$ and $\kappa>0$. A slightly more general discussion of the
spectra of the product of two self-adjoint operators can be found
in Appendix~\ref{sec:Eigenvalues-of-the product}.

\section{Convergence of iteration with an arbitrary initialization}

We now turn to iterative solution of equations \eqref{eq:auxiliary eq for phi}
and \eqref{eq:auxiliary eigenvalue problem Kohn-Sham-1}. We note
that iterations used in practical computations when properly initialized
do converge (see \cite{H-F-Y-G-B:2004,A-S-H-B:2019}); however, we
do not have a way to show that it will always be the case. Instead,
we consider solving the auxiliary eigenvalue problems for a fixed
parameter $\mu$ in \eqref{eq:auxiliary eigenvalue problem Kohn-Sham}
or $\kappa$ in \eqref{eq:auxiliary equation for psi}. Since the
operators are compact, the power iteration (combined with orthogonality
between selected eigenfunctions), will always converge. In the non-linear
Kohn-Sham and Hartree-Fock models, since components of the potential
depend on the eigenfunctions, we assume that convergence is not affected
by this nonlinear dependence and that solutions exists for all values
of the parameter. Note that while this is an assumption, any theory
that replaces the electron-electron interaction by an averaged field
becomes unusable if the solution of the resulting nonlinear eigenvalue
problem does not exist or cannot be obtained via an iteration. Once
the solution is obtained for a particular value of the parameter,
we can fix the potential and compute the derivative of the eigenvalue
with respect to the parameter. This derivative can then be used to
tune the parameter in order to arrive at $\lambda_{\mu}=1$. We provide
explicit expressions for the derivatives of the auxiliary eigenvalue
with respect to the parameter in Appendices~\ref{subsec:Appendix-C:-Monotone non-relativistic}
and \ref{subsec:Appendix-D:-Dependence of eig relativistic}. These
derivatives are derived assuming the potential is fixed.

As mentioned above, in practical computations we do not wait for the
full convergence of the power iteration; splitting this iteration
into two provides a way for understanding the reasons for convergence
and a way to achieve convergence if, for some reason, the practical
approach fails.

\subsection{Numerical demonstration}

The algorithm for computing bound states for relativistic equations
of QC is described in \cite{A-S-H-B:2019}. In order to demonstrate
the robust convergence regardless of starting condition, we consider
an example of using relativistic equations for the hydrogen atom,
which corresponds to a single atom with unit nuclear charge (i.e.,
$Z_{\alpha}=1$ in equation \ref{eq:nuclearpotential}). In all examples,
the non-relativistic energy parameter for the Green's function operator
is fixed at $-0.5$ (this is the non-relativistic ground-state energy
in the atomic units employed herein), with the corresponding relativistic
energy obtained by adding $mc^{2}$. For ease of comparison, the relativistic
energies reported in the text below are shifted by subtracting $mc^{2}$.
Also, the number of iterations is fixed at $100$ (we show only some
of them). In QC, the Dirac 4-spinor is interpreted as comprising two
2-spinors --- the ``large'' and ``small'' components, with the
terminology arising because, for the sought postive-energy states
that correspond to particles (in this case, electrons), the ``large''
component has significantly larger norm. For the free particle, the
large and small components of solutions to the Dirac equation are
related through the so-called kinetic balance condition

\[
\psi_{S}=\frac{\hbar}{2ic}\left(\sigma_{1}\frac{\partial}{\partial x_{1}}+\sigma_{2}\frac{\partial}{\partial x_{2}}+\sigma_{3}\frac{\partial}{\partial x_{3}}\right)\psi_{L}.
\]

Using MADNESS with wavelet order 8 and a domain width of 100.0 atomic
units, we examined four initial starting conditions as follows.
\begin{enumerate}
\item The standard starting guess of the non-relativistic hydrogen atom
solution in the first component of the large component with zero in
the second and the small component being determined from the kinetic
balance condition. The energies of the first two iterations are -0.500006490
and -0.500006270e-01, with the second iteration being converged to
all digits shown. Note that the exact Dirac-Coulomb energy for the
hydrogen atom is -0.500006656\ldots , but this is not obtained since
we have fixed the energy parameter in the integral operator at the
non-relativistic value.
\item The initial large and small components from the previous starting
guess are swapped, and the results shown in Figure \ref{fig:swappedconvergence}.
The non-relativistic ground state has zero-angular momentum (i.e.,
is an ``s'' function) and so, by construction, the small component
initial guess constructed by the kinetic-balance condition has unit
angular momentum (i.e., it is a ``p''function). These symmetries
are preserved in the relativistic solution. Hence, the initial guess
constructed by swapping the large and small components of the expected
non-relativistic initial guess in exact arithmetic is exactly orthogonal
to the sought solution. This is apparent in the iteration as displayed
in Figure~\ref{fig:swappedconvergence} --- the projection onto
the exact solution starts at about machine precision (being literally
numerical noise) and increases geometrically (circa 2x per iteration)
until it reaches circa 0.3, whereupon it converges rapidly to one.
The energy starts off large and negative (-1.9e4) but rapidly becomes
positive and decays to close to zero (presumably dominated by a superposition
of unbound electronic states). The energy stays near zero for many
iterations until the projection upon the exact solution approaches
0.1, at which point the energy converges rapidly to the desired electronic
ground state.
\item The third test employed a random initial guess in which function values
at the Gauss-Legendre quadrature points at 3 levels of refinement
in each dimension were set to a random value sampled uniformly in
$[0,10]$ with the resulting function multiplied by a characteristic
function to ensure it was zero on the edge of the computational volume
to satisfy the free-space boundary conditions. Different random functions
were used for each of the four components of the spinor. The energies
of the first three iterations were -1.88e4, -0.04, -0.33, with convergence
to nine significant figures of the energy being smoothly obtained
in 17 iterations.
\item The fourth test employed an initial guess that set one component of
the large component to a spherical Gaussian with exponent 1e8 (i.e.,
a very high energy, unbound, electronic state) and the corresponding
small component as determined by kinetic balance. The energies of
the first 4 iterations were (3.1e4, 9.5e4, 3.2e2, and -0.40 respectively),
and subsequently converged smoothly to nine significant figures of
the energy c in 11 iterations overall.
\end{enumerate}
We observe that independently of the initial guess, the iteration
converges as expected. The only difference is, naturally, in the number
of iterations needed to achieve convergence.

\begin{figure}
\begin{centering}
\includegraphics{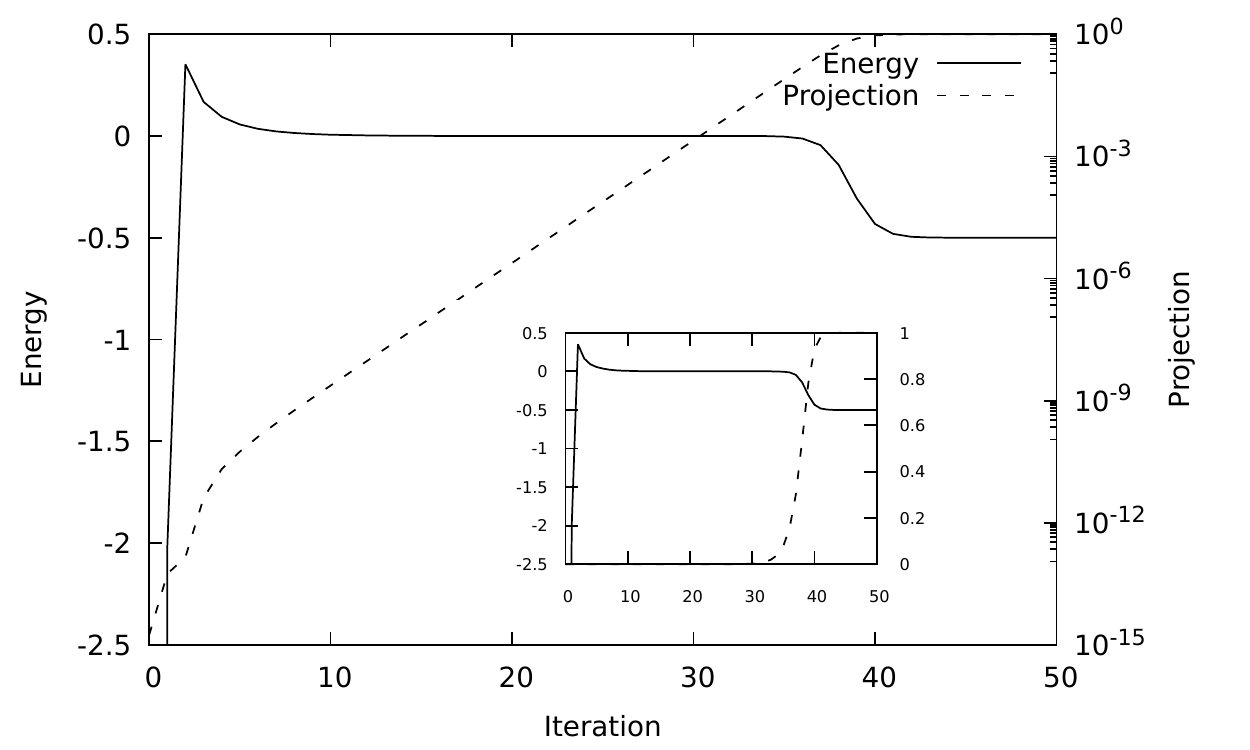}
\par\end{centering}
\caption{Convergence of the iteration starting from an initial guess with the
large and small components of the non-relativistic initial guess swapped
(i.e., one of the small components is set to the non-relativistic
solution of the hydrogen atom and the corresponding large component
is set to what would usually be chosen for the small component to
satisfy kinetic balance). As discussed in the text, in exact arithmetic
this initial guess would be orthogonal to the desired solution. The
inset plot differs only in that it does not use a log scale for the
projection onto the exact solution in order to reveal details close
to convergence. The energy of the first iteration (-3.8e4) is omitted
for clarity.\label{fig:swappedconvergence}}
\end{figure}

\section{Appendix}

\subsection{\label{subsec:Appendix-A:-Compact A^*A}Appendix A: Compact $A^{*}A$
implies $A$ is a compact operator}
\begin{lem}
\label{thm:A*A}If $A^{*}A$ is a compact operator then $A$ is also
a compact operator.
\end{lem}

\begin{proof}
We have 
\[
\left\Vert Ax\right\Vert ^{2}=\langle Ax,Ax\rangle=\langle A^{*}Ax,x\rangle\le\left\Vert A^{*}Ax\right\Vert \left\Vert x\right\Vert 
\]
and consider a bounded sequence $\left\{ x_{n}\right\} $, $\left\Vert x_{n}\right\Vert \le M$.
Since $A^{*}A$ is a compact operator, there exist a convergent subsequence
\[
A^{*}Ax_{n_{k}}
\]
which is then a Cauchy sequence. This implies that $Ax_{n_{k}}$ is
also a Cauchy sequence since, using the inequality above, we have
\[
\left\Vert Ax_{n_{k}}-Ax_{n_{l}}\right\Vert ^{2}=\left\Vert A\left(x_{n_{k}}-x_{n_{l}}\right)\right\Vert ^{2}\le2M\left\Vert A^{*}Ax_{n_{k}}-A^{*}Ax_{n_{l}}\right\Vert .
\]
Therefore, the subsequence $Ax_{n_{k}}$ is convergent and $A$ is
a compact operator.
\end{proof}
\begin{cor}
\label{cor:Corollary compact}If ~$A$ is a self-adjoint operator
and $A^{2}$ is compact then $A$ is also a compact operator.
\end{cor}

\subsection{\label{sec:Eigenvalues-of-the product}Appendix B: Eigenvalues of
the product of two self-adjoint operators}

We have
\begin{lem}
Let $A$ be a self-adjoint operator and $B$ a positive (or a negative)
definite self-adjoint operators. Then the spectrum of $AB$ is real.
\end{lem}

Consider eigenvalue problem 
\[
ABx=\lambda x,\,\:x\ne0.
\]

\begin{proof}
We have
\begin{equation}
\left\langle BABx,x\right\rangle =\left\langle ABx,Bx\right\rangle =\lambda\left\langle x,Bx\right\rangle ,\label{eq:relation for product of self-adjoints}
\end{equation}
and observe that $\left\langle ABx,Bx\right\rangle $ is real since
for any $y$, $\left\langle Ay,y\right\rangle =\left\langle y,Ay\right\rangle =\overline{\left\langle Ay,y\right\rangle }$.
Also for $x\ne0$, $\left\langle x,Bx\right\rangle =\left\langle Bx,x\right\rangle >0$
since $B$ is a positive self-adjoint operator (less than zero if
negative definite). We conclude that $\lambda$ is real.
\end{proof}
\begin{rem}
Let $A$ and $B$ be self-adjoint and, as before, we arrive at (\ref{eq:relation for product of self-adjoints}).
We have $\left\langle BABx,x\right\rangle $ is real and $\left\langle x,Bx\right\rangle $
is real. We conclude that $\lambda$ is real provided $\left\langle x,Bx\right\rangle \ne0$.
This requirement is satisfied if $B$ is a positive or a negative
definite operator; a weaker assumption is that the eigenvectors of
$AB$ must be such that $\left\langle x,Bx\right\rangle \ne0$. As
a simple example consider 
\[
A=\left(\begin{array}{cc}
0 & 1\\
1 & 0
\end{array}\right)\,\,\,\mbox{and}\,\,\,\,B=\left(\begin{array}{cc}
1 & 0\\
0 & -1
\end{array}\right).
\]
It is easy to check that the eigenvalues of $AB$ in this case are
pure imaginary ($i$ and $-i$) and the eigenvectors of $AB$ are
such that $\left\langle x,Bx\right\rangle =0$. On the other hand,
if 
\[
A=\left(\begin{array}{ccc}
0 & 1 & 0\\
1 & 0 & 0\\
0 & 0 & 1
\end{array}\right)\,\,\,\mbox{and}\,\,\,\,B=\left(\begin{array}{ccc}
1 & 0 & 0\\
0 & 1 & 0\\
0 & 0 & -1
\end{array}\right),
\]
then the eigenvalues of $A$ are $\{-1,1,1\}$, the eigenvalues of
\textbf{$B$} $\{1,1,-1\}$ so that neither $A$ or $B$ are positive/negative
definite. However the eigenvalues of $AB$ are $\{-1,-1,1\}$ and
it is easy to check that the eigenvectors of $AB$ satisfy $\left\langle x,Bx\right\rangle \ne0$.%
\end{rem}

\begin{lem}
Let $A$ be a bounded and $B$ a compact positive (or a negative)
definite self-adjoint operators. Then $AB$ has only discrete real
eigenvalues with the the only possible accumulation point at zero.
\end{lem}

\begin{proof}
Let us consider $B^{1/2}AB^{1/2}$. Since $B$ is positive definite,
$B^{1/2}$ exist and is self-adjoint. According to Corollary~\ref{cor:Corollary compact}
$B^{1/2}$ is a compact operator and, therefore, $B^{1/2}AB^{1/2}$
is also compact. Since it is self-adjoint, the spectral theorem for
compact operators in a Hilbert space is applicable and we know that
its spectrum is discrete, with the only possible accumulation point
at zero.  Since $AB=B^{-1/2}\left(B^{1/2}AB^{1/2}\right)B^{1/2}$,
it has the same spectrum as $B^{1/2}AB^{1/2}$.
\end{proof}

\subsection{Appendix C: \label{subsec:Appendix-C:-Monotone non-relativistic}Monotone
dependence of eigenvalue on the parameter in non-relativistic case}

Let us consider 

\begin{equation}
\lambda_{\mu}\psi_{\mu}=-G_{\mu}V\psi_{\mu},\,\,\,,\lambda_{\mu}>0,\label{eq:auxillary eig problem for psi , non-relativistic}
\end{equation}
or 
\begin{equation}
\lambda_{\mu}\phi_{\mu}=-G_{\mu}^{1/2}VG_{\mu}^{1/2}\phi_{\mu},\,\,\,\left\Vert \phi_{\mu}\right\Vert =1,\,\,\,,\lambda_{\mu}>0,\label{eq:auxillary eig problem for psi , non-relativistic-1}
\end{equation}
where $G_{\mu}$ is the Green's function 
\[
G_{\mu}\left(\mathbf{p}\right)=\frac{1}{\mu^{2}+\left\Vert \mathbf{p}\right\Vert ^{2}},
\]
in the momentum space and $\phi_{\mu}=G_{\mu}^{-1/2}\psi_{\mu}$.
Here we assume that the potential $V$ does not depend on $\mu$.
Since within the iteration potential does depend on $\mu$, our conclusion
is applicable once the iteration converged so that the potential can
be fixed.

We use 
\[
\frac{dG_{\mu}}{d\mu}=-2\mu G_{\mu}^{2}
\]
in our derivations below.
\begin{lem}
\label{lem:If-the-potential}If the potential $V$ does not depend
on the parameter $\mu$, then the derivative of an eigenvalue with
respect to the parameter $\mu$ is negative,
\begin{equation}
\frac{d\lambda_{\mu}}{d\mu}=-2\mu\lambda_{\mu}\left\Vert G_{\mu}^{1/2}\phi_{\mu}\right\Vert .\label{derivative of lambda}
\end{equation}
\end{lem}

\begin{proof}
Differentiating (\ref{eq:auxillary eig problem for psi , non-relativistic}),
we have
\[
\frac{d\lambda_{\mu}}{d\mu}\psi_{\mu}+\lambda_{\mu}\frac{d\psi_{\mu}}{d\mu}=2\mu G_{\mu}^{2}V\psi_{\mu}-G_{\mu}V\frac{d\psi_{\mu}}{d\mu},
\]
and computing the inner product with $G_{\mu}^{-1}\psi_{\mu}$, obtain
\[
\frac{d\lambda_{\mu}}{d\mu}\left\langle \psi_{\mu},G_{\mu}^{-1}\psi_{\mu}\right\rangle +\lambda_{\mu}\left\langle \frac{d\psi_{\mu}}{d\mu},G_{\mu}^{-1}\psi_{\mu}\right\rangle =2\mu\left\langle G_{\mu}^{2}V\psi_{\mu},G_{\mu}^{-1}\psi_{\mu}\right\rangle -\left\langle G_{\mu}V\frac{d\psi_{\mu}}{d\mu},G_{\mu}^{-1}\psi_{\mu}\right\rangle .
\]
Since

\[
\lambda_{\mu}G_{\mu}^{-1}\psi_{\mu}=-V\psi_{\mu}
\]
and $G_{\mu}$ and $V$ are symmetric operators, we have
\[
\frac{d\lambda_{\mu}}{d\mu}\left\langle \psi_{\mu},G_{\mu}^{-1}\psi_{\mu}\right\rangle =2\mu\left\langle G_{\mu}^{2}V\psi_{\mu},G_{\mu}^{-1}\psi_{\mu}\right\rangle 
\]
or, using (\ref{eq:auxillary eig problem for psi , non-relativistic}),
\[
\frac{d\lambda_{\mu}}{d\mu}\left\langle \psi_{\mu},G_{\mu}^{-1}\psi_{\mu}\right\rangle =-2\mu\lambda_{\mu}\left\langle G_{\mu}\psi_{\mu},G_{\mu}^{-1}\psi_{\mu}\right\rangle =-2\mu\lambda_{\mu}\left\Vert \psi_{\mu}\right\Vert ^{2}
\]
yielding (\ref{derivative of lambda}).
\end{proof}

\subsection{Appendix D: \label{subsec:Appendix-D:-Dependence of eig relativistic}Dependence
of eigenvalue on the parameter in the relativistic case}

Next we consider the auxiliary eigenvalue problem for the Dirac equation
\begin{equation}
\lambda\left(\kappa\right)\psi\left(\kappa\right)=-\frac{1}{\hbar^{2}c^{2}}\left(\mathcal{H}_{0}+E\left(\kappa\right)\mathcal{I}\right)\mathcal{G}\left(\kappa\right)\mathcal{V}\psi\left(\kappa\right),\label{eq:auxiliary eig problem for Dirac}
\end{equation}
where $E<mc^{2}$
\[
E=\left(m^{2}c^{4}-\kappa^{2}c^{2}\hbar^{2}\right)^{1/2},
\]
and 
\[
\frac{dE}{d\kappa}=-\kappa c^{2}\hbar^{2}\left(m^{2}c^{4}-\kappa^{2}c^{2}\hbar^{2}\right)^{-1/2}=-\kappa\frac{c^{2}\hbar^{2}}{E}.
\]
In (\ref{eq:auxiliary eig problem for Dirac}) 
\[
\mathcal{G}\left(\kappa\right)=G\left(\kappa\right)\mathcal{I},
\]
where $\mathcal{I}$ is the $4\times4$ identity matrix, $\mathcal{V}$
is $4\times4$ potential operator 
\[
\mathcal{V}=\left(\begin{array}{cccc}
V & 0 & 0 & 0\\
0 & V & 0 & 0\\
0 & 0 & V & 0\\
0 & 0 & 0 & V
\end{array}\right)
\]
and 
\[
\mathcal{H}_{0}=\left(\begin{array}{cc}
\sigma_{0}mc^{2} & \frac{\hbar c}{i}\left(\sigma_{1}\frac{\partial}{\partial x_{1}}+\sigma_{2}\frac{\partial}{\partial x_{2}}+\sigma_{3}\frac{\partial}{\partial x_{3}}\right)\\
\frac{\hbar c}{i}\left(\sigma_{1}\frac{\partial}{\partial x_{1}}+\sigma_{2}\frac{\partial}{\partial x_{2}}+\sigma_{3}\frac{\partial}{\partial x_{3}}\right) & -\sigma_{0}mc^{2}
\end{array}\right).
\]
In the momentum space 
\[
\widehat{\frac{1}{\hbar^{2}c^{2}}\left(\mathcal{H}_{0}+E\mathcal{I}\right)}=\left(\begin{array}{cccc}
\frac{m}{\hbar^{2}}+\frac{E}{\hbar^{2}c^{2}} & 0 & \frac{p_{3}}{\hbar c} & \frac{p_{1}-ip_{2}}{\hbar c}\\
0 & \frac{m}{\hbar^{2}}+\frac{E}{\hbar^{2}c^{2}} & \frac{p_{1}+ip_{2}}{\hbar c} & \frac{-p_{3}}{\hbar c}\\
\frac{p_{3}}{\hbar c} & \frac{p_{1}-ip_{2}}{\hbar c} & -\frac{m}{\hbar^{2}}+\frac{E}{\hbar^{2}c^{2}} & 0\\
\frac{p_{1}+ip_{2}}{\hbar c} & \frac{-p_{3}}{\hbar c} & 0 & -\frac{m}{\hbar^{2}}+\frac{E}{\hbar^{2}c^{2}}
\end{array}\right)
\]
and this operator commutes with $\mathcal{G}\left(\kappa\right)$.
\begin{lem}
If the potential $\mathcal{V}$ does not depend on the parameter $\mu$,
then for the derivative of an eigenvalue in (\ref{eq:auxiliary eig problem for Dirac}),
we have
\[
\frac{d\lambda\left(\kappa\right)}{d\kappa}=\kappa\lambda\left(\kappa\right)\frac{1}{\hbar^{2}c^{2}}\frac{1}{\left\langle \psi,\left(\mathcal{H}_{0}-E\mathcal{I}\right)\psi\right\rangle }\left(-\frac{1}{E\left(\kappa\right)}\left\Vert \mathcal{G}^{1/2}\mathcal{H}_{0}\psi\right\Vert ^{2}+E\left(\kappa\right)\left\Vert \mathcal{G}^{1/2}\psi\right\Vert ^{2}\right).
\]
\end{lem}

As in the non-relativistic case, the expression for the derivative
does not have explicit dependence on the potential. However, in this
case the derivative may not be sign definite. In both, non-relativistic
and relativistic cases, expressions for the derivative can be used
to set up a combination of the power iteration and the Newton method
to solve for $\lambda\left(\kappa\right)=1$ in order to compute the
bound states
\begin{proof}
We have
\begin{equation}
\frac{d\lambda\left(\kappa\right)}{d\kappa}\psi+\lambda\left(\kappa\right)\frac{d\psi}{d\kappa}=-\frac{1}{\hbar^{2}c^{2}}\frac{dE}{d\kappa}\mathcal{G}\mathcal{V}\psi-\frac{1}{\hbar^{2}c^{2}}\left(\mathcal{H}_{0}+E\mathcal{I}\right)\frac{d\mathcal{G}}{d\kappa}\mathcal{V}\psi-\frac{1}{\hbar^{2}c^{2}}\left(\mathcal{H}_{0}+E\mathcal{I}\right)\mathcal{G}\mathcal{V}\frac{d\psi}{d\kappa}\label{eq:differentiated equation}
\end{equation}
and evaluate the inner product of both sides of (\ref{eq:differentiated equation})
with $\mathcal{G}^{-1}\left(\mathcal{H}_{0}+E\mathcal{I}\right)^{-1}\psi$.
We first show that 
\[
\left\langle \lambda\left(\kappa\right)\frac{d\psi}{d\kappa},\mathcal{G}^{-1}\left(\mathcal{H}_{0}+E\mathcal{I}\right)^{-1}\psi\right\rangle =-\left\langle \frac{1}{\hbar^{2}c^{2}}\left(\mathcal{H}_{0}+E\mathcal{I}\right)\mathcal{G}\mathcal{V}\frac{d\psi}{d\kappa},\mathcal{G}^{-1}\left(\mathcal{H}_{0}+E\mathcal{I}\right)^{-1}\psi\right\rangle .
\]
Applying $\mathcal{G}^{-1}\left(\mathcal{H}_{0}+E\mathcal{I}\right)^{-1}$
on both sides of (\ref{eq:auxiliary eig problem for Dirac}), we have
\[
\lambda\left(\kappa\right)\mathcal{G}^{-1}\left(\kappa\right)\left(\mathcal{H}_{0}+E\left(\kappa\right)\mathcal{I}\right)^{-1}\psi\left(\kappa\right)=-\frac{1}{\hbar^{2}c^{2}}\mathcal{V}\psi\left(\kappa\right)
\]
so that 
\[
\left\langle \lambda\left(\kappa\right)\frac{d\psi}{d\kappa},\mathcal{G}^{-1}\left(\mathcal{H}_{0}+E\mathcal{I}\right)^{-1}\psi\right\rangle =-\frac{1}{\hbar^{2}c^{2}}\left\langle \frac{d\psi}{d\kappa},\mathcal{V}\psi\right\rangle .
\]
Since $\mathcal{H}_{0}+E\mathcal{I}$ and $\mathcal{G}$ commute,
we also have 
\[
-\left\langle \frac{1}{\hbar^{2}c^{2}}\left(\mathcal{H}_{0}+E\mathcal{I}\right)\mathcal{G}\mathcal{V}\frac{d\psi}{d\kappa},\mathcal{G}^{-1}\left(\mathcal{H}_{0}+E\mathcal{I}\right)^{-1}\psi\right\rangle =-\frac{1}{\hbar^{2}c^{2}}\left\langle \frac{d\psi}{d\kappa},\mathcal{V}\psi\right\rangle .
\]
We thus obtain
\begin{eqnarray}
\frac{d\lambda\left(\kappa\right)}{d\kappa}\left\langle \psi,\mathcal{G}^{-1}\left(\mathcal{H}_{0}+E\mathcal{I}\right)^{-1}\psi\right\rangle  & = & -\frac{1}{\hbar^{2}c^{2}}\frac{dE}{d\kappa}\left\langle \mathcal{G}\mathcal{V}\psi,\mathcal{G}^{-1}\left(\mathcal{H}_{0}+E\mathcal{I}\right)^{-1}\psi\right\rangle +\nonumber \\
 &  & \frac{2\kappa}{\hbar^{2}c^{2}}\left\langle \left(\mathcal{H}_{0}+E\mathcal{I}\right)\mathcal{G}^{2}\mathcal{V}\psi,\mathcal{G}^{-1}\left(\mathcal{H}_{0}+E\mathcal{I}\right)^{-1}\psi\right\rangle .\label{eq:intermediate eq}
\end{eqnarray}
Using 
\[
-\frac{1}{\hbar^{2}c^{2}}\mathcal{G}\mathcal{V}\psi=-\frac{1}{\hbar^{2}c^{2}}\left(\mathcal{H}_{0}+E\mathcal{I}\right)^{-1}\left(\mathcal{H}_{0}+E\mathcal{I}\right)\mathcal{G}\mathcal{V}\psi=\left(\mathcal{H}_{0}+E\mathcal{I}\right)^{-1}\psi,
\]
\[
\left(\mathcal{H}_{0}+E\mathcal{I}\right)^{-1}=\frac{1}{\hbar^{2}c^{2}}\mathcal{G}\left(\mathcal{H}_{0}-E\mathcal{I}\right),
\]
and
\[
\mathcal{G}^{-1}\left(\mathcal{H}_{0}+E\mathcal{I}\right)^{-1}=\frac{1}{\hbar^{2}c^{2}}\left(\mathcal{H}_{0}-E\mathcal{I}\right),
\]
we have for the first term on the right hand side of (\ref{eq:intermediate eq})
\begin{eqnarray*}
-\frac{1}{\hbar^{2}c^{2}}\frac{dE}{d\kappa}\left\langle \mathcal{G}\mathcal{V}\psi,\mathcal{G}^{-1}\left(\mathcal{H}_{0}+E\mathcal{I}\right)^{-1}\psi\right\rangle  & = & -\kappa\lambda\left(\kappa\right)\frac{c^{2}\hbar^{2}}{E}\left\langle \left(\mathcal{H}_{0}+E\mathcal{I}\right)^{-1}\psi,\mathcal{G}^{-1}\left(\mathcal{H}_{0}+E\mathcal{I}\right)^{-1}\psi\right\rangle \\
 & = & -\kappa\lambda\left(\kappa\right)\frac{1}{E}\frac{1}{\hbar^{2}c^{2}}\left\langle \mathcal{G}\left(\mathcal{H}_{0}-E\mathcal{I}\right)\psi,\left(\mathcal{H}_{0}-E\mathcal{I}\right)\psi\right\rangle \\
 & = & -\kappa\lambda\left(\kappa\right)\frac{1}{E}\frac{1}{\hbar^{2}c^{2}}\left(\left\langle \mathcal{G}\mathcal{H}_{0}\psi,\mathcal{H}_{0}\psi\right\rangle -2E\left\langle \mathcal{G}\mathcal{H}_{0}\psi,\psi\right\rangle +E^{2}\left\langle \mathcal{G}\psi,\psi\right\rangle \right)\\
 & = & -\kappa\lambda\left(\kappa\right)\frac{1}{E}\frac{1}{\hbar^{2}c^{2}}\left(\left\Vert \mathcal{G}^{1/2}\mathcal{H}_{0}\psi\right\Vert ^{2}-2E\left\langle \mathcal{G}\mathcal{H}_{0}\psi,\psi\right\rangle +E^{2}\left\Vert \mathcal{G}^{1/2}\psi\right\Vert ^{2}\right),
\end{eqnarray*}
and, for the second,
\begin{eqnarray*}
\frac{2\kappa}{\hbar^{2}c^{2}}\left\langle \left(\mathcal{H}_{0}+E\mathcal{I}\right)\mathcal{G}^{2}\mathcal{V}\psi,\mathcal{G}^{-1}\left(\mathcal{H}_{0}+E\mathcal{I}\right)^{-1}\psi\right\rangle  & = & -2\kappa\lambda\left(\kappa\right)\left\langle \mathcal{G\psi},\mathcal{G}^{-1}\left(\mathcal{H}_{0}+E\mathcal{I}\right)^{-1}\psi\right\rangle \\
 & = & -2\kappa\lambda\left(\kappa\right)\left\langle \mathcal{\psi},\left(\mathcal{H}_{0}+E\mathcal{I}\right)^{-1}\psi\right\rangle \\
 & = & -2\kappa\lambda\left(\kappa\right)\frac{1}{\hbar^{2}c^{2}}\left\langle \mathcal{\psi},\mathcal{G}\left(\mathcal{H}_{0}-E\mathcal{I}\right)\psi\right\rangle \\
 & = & -2\kappa\lambda\left(\kappa\right)\frac{1}{\hbar^{2}c^{2}}\left(\left\langle \mathcal{\psi},\mathcal{G}\mathcal{H}_{0}\psi\right\rangle -E\left\Vert \mathcal{G}^{1/2}\psi\right\Vert ^{2}\right).
\end{eqnarray*}
Also using 
\[
\left\langle \psi,\mathcal{G}^{-1}\left(\mathcal{H}_{0}+E\mathcal{I}\right)^{-1}\psi\right\rangle =\frac{1}{\hbar^{2}c^{2}}\left\langle \psi,\left(\mathcal{H}_{0}-E\mathcal{I}\right)\psi\right\rangle ,
\]
we obtain from (\ref{eq:intermediate eq})
\[
\frac{d\lambda\left(\kappa\right)}{d\kappa}\left\langle \psi,\left(\mathcal{H}_{0}-E\mathcal{I}\right)\psi\right\rangle =\kappa\lambda\left(\kappa\right)\frac{1}{\hbar^{2}c^{2}}\left(-\frac{1}{E\left(\kappa\right)}\left\Vert \mathcal{G}^{1/2}\mathcal{H}_{0}\psi\right\Vert ^{2}+E\left(\kappa\right)\left\Vert \mathcal{G}^{1/2}\psi\right\Vert ^{2}\right).
\]
\end{proof}

\bibliographystyle{plain}


\end{document}